\documentclass[final,leqno,onefignum,onetabnum]{siamltex1213}
\usepackage{amsmath,amsfonts,amscd,eucal,latexsym,mathrsfs,epsfig}

\newcommand{\R}{{\mathbb{R}^3}}
\newcommand{\bv}{{\bf v}}
\newcommand{\bk}{{\bf k}}
\newcommand{\bvun}{{{\bf v}_1}}

\newcommand{\vmv}{{\bf v}-{\bf v}_1}

\title{INTERPLAY OF BOLTZMANN EQUATION AND CONTINUITY EQUATION FOR ACCELERATED ELECTRONS IN SOLAR FLARES\footnotemark[1]}
\author{A. CODISPOTI\footnotemark[2] \and N. PINAMONTI\footnotemark[2] \footnotemark[3]}

\begin{document} 
\renewcommand{\thefootnote}{\fnsymbol{footnote}}
 
\maketitle

\footnotetext[1]{Received by...}
\footnotetext[2]{Dipartimento di Matematica, Universit\`a di Genova,  via Dodecaneso, 35 I-16146 Genova,  Italy\\ (codispoti@dima.unige.it, pinamont@dima.unige.it).}
\footnotetext[3]{INFN, Sezione di Genova, Via Dodecaneso, 33 I-16146 Genova, Italy.}
\renewcommand{\thefootnote}{\arabic{footnote}} 

\begin{abstract}
During solar flares a large amount of electrons are accelerated within the plasma present in the solar atmosphere. Accurate measurements of the motion of these electrons start becoming available from the analysis of hard X-ray imaging-spectroscopy observations. 
In this paper, we discuss the linearized perturbations of the Boltzmann kinetic equation describing an ensemble of electrons accelerated by the energy release occurring during solar flares. Either in the limit of high energy or at vanishing background temperature such an equation reduces to a continuity equation equipped with an extra force of stochastic nature.
This stochastic force is actually described by the well known energy loss rate due to Coulomb collision with ambient particles, but, in order to match the collision kernel in the linearized Boltzmann equation it needs to be treated in a very specific manner. 
In the second part of the paper the derived continuity equation is solved with some hyperbolic techniques, and the obtained solution is written in a form suitable to be compared with data gathered by hard X--ray imaging--spectroscopy telescopes. 
Finally, a first validation of the model with NASA Reuven Ramaty High Energy Solar Spectroscopic Imager spectrometer measurements is provided.
\end{abstract}

\begin{keywords} 
Solar flares, Boltzmann equation, Landau equation, continuity equation, electron transport.
\end{keywords}

\begin{AMS} 
76P05, 82D10, 85A25
\end{AMS} 

\section{Introduction}As seen from the Earth the Sun appears quiet and stable. However, complicated phenomena can be observed, looking closer at its surface.
Among all perturbations occurring in the solar atmosphere, solar flares are the most energetic ones. 
Actually, during solar flares, huge quantities of energy are released. 
In the case of arc--shaped flares, according to the standard picture \cite{brown71, fletcher,tan88},
a large number of electrons are first accelerated above the loop top and 
then injected into the loop, finally they move down towards the chromosphere
following the magnetic field lines.
During this motion, electrons lose energy because of various physical processes. The most efficient energy loss mechanisms are Coulomb collisions with ambient electrons present in the solar atmosphere \cite{brown72,lh76,em78}.
However, these kind of interactions are not perfectly elastic and, in every collision, a small, almost negligible, part of energy is radiated away due to Bremsstrahlung, mainly in the Hard X--ray domain \cite{petrosian73, kon11}.
Photons emitted by means of this mechanism are visible also far away from the solar surface. 
Thanks to this effect we are able in principle to indirectly detect electron densities by analyzing hard X--ray photons emitted during flares.

In the X--regime, solar flares are observed by satellite imaging--spectrometers such as, YOHKOH \cite{acton92}, RHESSI \cite{ldh02}, which is still in operation, and STIX in the future Solar Orbiter mission \cite{benz12}.
Since these instruments are capable of collecting data both spatially and spectrally resolved, precise images of photon fluxes at various energies can be constructed. In particular, in  \cite{pmh07}, Piana et al showed how to construct maps of electron fluxes $F({\bf x},E)$ averaged over the line of sight and thus information about electrons (electron maps) can be obtained from satellite spectrometer data. 
Nowadays, these kind of measurements are becoming very precise and permit to test refined models of electron transport mechanisms.

Unfortunately, the information about the direction of electrons velocities is lost in the detection procedure, because only their energy can be measured. We shall see that this issue slightly modifies the validation process of theoretical models by comparison with the measured data (i.e., the images provided by these satellites).

 From the theoretical point of view, it is believed that during solar flares electrons are accelerated by magnetic reconnection. However, the actual mechanism governing this acceleration is still under investigation \cite{asch02,brown09,petrosian,zar11,bian12,guo13}. For this reason several statistical mechanisms of electron acceleration and various theoretical approaches have been developed in order to model the physics of solar flares.

Usually, since electrons are indistinguishable particles, their ensemble is described by a {\em density function} $f(t,{\bf x},\bv)$ which expresses the number of particles with velocity $\bv$ at the time $t$ in the position ${\bf x}$. The integral of $f$ over all possible velocities gives the standard particle density $n(t,{\bf x})$.
The evolution of the density function is given by various equations of plasma physics adapted to the kind of plasma present in the solar atmosphere. 
As an example, we recall the analysis presented in \cite{kennel-engelmann, schlickeiser-1, schlickeiser-2}, where perturbations of charged particles in a plasma are shown to satisfy a Fokker Planck equation.
Recently several theoretical diffusion models were considered and adapted to solar flares e.g. \cite{petrosian,bian12}.

Among all of the various possibilities, the continuity equation adapted to the phase space was employed to describe the electron motion along magnetic field lines in solar flares. 
We recall that the standard continuity equation has the following form

\begin{equation}\label{eq:continuity-total-derivative}
\frac{df}{dt}=\frac{\partial  f}{\partial t} + 
\bv \cdot \nabla_{\bf x}   f+ 
\frac{{\bf F}}{m}\cdot \nabla_\bv f 
=  S
\end{equation}
where a source term $S$ is considered to model injection of electrons in the flaring region due to magnetic reconnection. Furthermore, ${\bf F}$ is the total external force acting on electrons and $m$ is the electron mass. 
Notice in particular that whenever the force ${\bf F}$ is positional the previous equation is equivalent to 
\begin{equation}\label{eq:continuity-total-derivative-F}
\frac{\partial  f}{\partial t} + 
\nabla_{\bf x}   \cdot  \left(  \bv  f\right)+ 
\nabla_\bv\cdot\left(   \frac{{\bf F}}{m} f \right)
=  S.
\end{equation}
However, in order to take into account the statistical energy loss due to Coulomb collisions, the previous two equations are modified by adding the average value of a  stochastic force ${\bf F}_s$ to the total external force  ${\bf F}$. 
In particular ${\bf F}_s$ is chosen to match the {\bf energy loss rate} $\langle \nabla_{\bf x} E \rangle$ computed as the mean value of the energy loss in collisions with ambient particles, as e.g. \cite{em78, sp62}.
With this modification the continuity equation \eqref{eq:continuity-total-derivative} takes the form
\begin{equation}\label{eq:continuity-stochastic}
\frac{\partial  f}{\partial t} + 
\bv \cdot \nabla_{\bf x}   f+ 
\frac{{\bf F}}{m}\cdot \nabla_\bv f +
\frac{\langle\nabla_{\bf x} E\rangle}{m}\cdot \nabla_\bv f
=  S\;.
\end{equation}

A slightly different equation, valid for accelerated electrons propagating in a solar plasma, has been derived by Emslie, Barrett and Brown \cite{emsbarbro01} as a consequence of conservation of the electron flux along every energy channel.
It has been widely used in order to understand some details about solar flares e.g. \cite{brown71,brown72,syrshe72}.
Furthermore, it has recently been employed for the validation of various statistical energy loss rate models \cite{torre, codi13}.

However, in general the stochastic term ${\bf F}_s=\langle \nabla_{\bf x} E \rangle$ is not positional, it explicitly depends on 
the velocities, hence equation

\eqref{eq:continuity-stochastic} is not anymore equivalent to 
\begin{equation}\label{eq:continuity-stochastic-wrong}
\frac{\partial  f}{\partial t} + 
\nabla_{\bf x}   \cdot  \left(  \bv  f\right)+ 
\nabla_\bv\cdot\left(   \frac{{\bf F}}{m} f \right)
+
\nabla_\bv\cdot\left(   \frac{\langle\nabla_{\bf x} E\rangle}{m} f \right)
=  S
\end{equation}
contrary to the case of purely positional external forces.
In this paper, we show that only continuity equation in the form \eqref{eq:continuity-stochastic}
matches the asymptotic form of the Boltzmann equation   
which  is a non--linear  kinetic equation 
 that takes into account interaction/collisions of particles by means of a non linear collision kernel.
In particular, we shall obtain the stochastic contribution 
in the form
\begin{equation}\label{eqn:contloss}
L_C(f)(t,{\bf x},\bv)=\frac{\langle\nabla_{\bf x} E\rangle}{m} \cdot \nabla_\bv  f
\end{equation}
as an asymptotic approximation of the Boltzmann collision kernel.
Furthermore,  since the physical system we have in mind is open, 
the equation from which we start is the Boltzmann equation equipped with some appropriate given source term $S'$, which is related to $S$ in \eqref{eq:continuity-total-derivative},
\begin{equation}\label{eq:boltzmann}
\frac{\partial  f}{\partial t} + 
 \bv \cdot  \nabla_{\bf x} f+ 
\frac{{\bf F}}{m} \cdot \nabla_\bv f  =  Q(f,f) +S'
\end{equation}
where $Q(f,f)$ is the non linear Boltzmann collision kernel for Coulomb collisions and where, as before, ${\bf F}$ are the external classical forces.
In this paper we  treat $Q(f,f)$ in the Landau approximation \cite{degluc92}, namely in the limit of grazing collisions, or in the limit of large values for the Coulomb logarithm. 
Furthermore, we consider linear perturbations $\delta f$ over solutions $f$ of the Boltzmann equation.
In the asymptotic regime, i.e. for large velocities $|\bv|$ or for small temperatures in the case of Maxwell background,
the perturbation $\delta f$ satisfies  an equation  of the form \eqref{eq:continuity-stochastic} provided the stochastic force is taken into account with a term of the form $\frac{\langle\nabla_{\bf x} E\rangle}{m} \cdot \nabla_\bv  f$, which is slightly different from that considered in \cite{torre, codi13}.

After having obtained the continuity equation \eqref{eq:continuity-stochastic} equipped with the stochastic force used to model Coulomb collisions, we shall discuss how its outcome can be compared with measurements. As a first step, we have to adapt the obtained continuity equation for electron densities resolved in position and velocities, to describe electron fluxes resolved in position and energy. 
In order to accomplish this task, an averaging procedure over the velocity directions needs to be employed. 
Operating in this way we obtain a new continuity equation and the study of the differences between this new model and the original one in \cite{emsbarbro01} is at the core of this paper. We will address such study from both a theoretical viewpoint and by means of a computational validation based on empirical data measured by RHESSI.
The paper is organized as follows. In the next section we shall recall some standard facts about the Landau approximation of the Boltzmann collision kernel for Coulomb interactions and its linearization.
In the third section we shall analyze the asymptotic form of the linearized Boltzmann equation and we establish the connection to the continuity equation employed in the literature.
In the fourth section we shall discuss the solution of the continuity equation in the case of isotropic perturbations. We shall in particular highlight the hyperbolic character of the obtained equation.
In the fifth section we shall compare the obtained continuity equation for electron fluxes with the one presented in \cite{emsbarbro01} also discussing their  validation with measured data. Hence, we shall show how the proposed modified equation is consistent with experimental data provided by imaging spectroscopy. Finally, some comments are drawn in the last section.

\section{Landau approximation of the Boltzamnn collision kernel and its linearization}

The collision kernel present in the Boltzmann equation \eqref{eq:boltzmann} takes into account elastic collisions among 
particles distributed according to some density function $f(t,{\bf x},\bv)$. 
Furthermore, since collisions are supposed to be elastic, the 
   momentum and the kinetic energy are conserved in every collision. It means that, if we indicate 
  by ${\bf v}$ and ${\bf v}_1$ the interacting particles  velocities 
  before the collision, and by ${\bf v}'$ and ${\bf v}_1'$ the interacting particles velocities after the 
  collision, the following relations hold
  \begin{equation*}
{\bf v}+{\bf v}_1={\bf v}'+{\bf v}_1',\qquad 
|{\bf v}|^2+|{\bf v}_1|^2 =|{\bf v}'|^2+|{\bf v}_1'|^2.
  \end{equation*}
  From these equations, indicating by $\sigma=(\sin\theta\cos\phi,\sin\theta\sin\phi,\cos\theta)$  the 
  deviation angle, we can write explicitly ${\bf v}'$ and ${\bf v}_1'$ as
 \begin{equation*}
{\bf v}'=\frac{{\bf v}+{\bf v}_1}{2}+\frac{|{\bf v}-{\bf v}_1|}{2}\sigma,
\qquad
{\bf v}'_1=\frac{{\bf v}-{\bf v}_1}{2}-\frac{|{\bf v}-{\bf v}_1|}{2}\sigma.
 \end{equation*} 
The collision kernel $Q(f,f)$, which  accounts the interactions between colliding particles, takes the form \cite{glas96}
\begin{equation}\label{eq:boltzmann-kernel}
Q(f,f)(t,{\bf x},{\bf v})=\int_{\mathbb{R}^3}\int_{S^2} B(|{\bf v}-{\bf v}_1|,\sigma)(f'f'_1-ff_1)d^2\sigma d^3\bv_1
\end{equation}
where, in order to shorten the notation, it is understood that  $f'=f(\bv')$, $f_1=f(\bv_1)$ and $f'_1=f(\bv_1')$.
Furthermore, in the case of Coulomb collisions, the integral kernel $B$ is related to the Rutherford cross section
 \begin{equation}\label{eqn:Ruth}
   B(|{\bf v}-{\bf v}_1|,\sigma)=\frac{e^4}{m^2}\frac{1}{|{\bf v}-{\bf v}_1|^3 \sin^4\frac{\theta}{2}}\quad\quad \theta\geq\theta_0
 \end{equation}
where  $\theta_0$ is a cutoff representing the minimum scattering angle, which is related to the Debye screening length, and to the Coulomb logarithm $\Lambda$ 
 \[
 \Lambda=-\log{\sin{\frac{\theta_0}{2}}}\approx -\log{\frac{\theta_0}{2}}.
 \]
 The typical values of the Coulomb logarithm for the plasma in the solar atmosphere are sufficiently large to justify the assumption of grazing collisions, namely to consider only small scattering angles \cite{alevil2004}.  
As a rather direct consequence, we obtain an approximate expansion $Q_L$ of $Q$ through a first order expansion in the scattering angle and the evaluation of the integral over all possible scattering directions.
In this way the Landau approximation  of the Boltzmann collision kernel is obtained.

The Landau equation was derived for the first time by Lev Davidovi\v{c} Landau \cite{lanlif81}.
The limits in which this equation represents a good approximation of the Boltzmann equation and the relations between the solutions of the two equations have been widely analyzed \cite{degluc92,des92, bucor99, linxio14}.
In particular, the  solutions of the Boltzmann equation are well approximated by solutions of
the Landau equation if the density function $f$ is spatially homogenous \cite{arsbur91,gou97,vil98}. 
The analogous result in the non homogeneous case has been rather recently obtained in \cite{alevil2004}.

The Landau approximation of the Boltzmann collision kernel, which is also called Fokker Planck collision kernel \cite{degluc92}, 
has the form 
\begin{equation}
  Q_L(f,f)(t,{\bf x},{\bf v})=\frac{8 \pi\Lambda e^4}{m^2} \nabla \cdot\int_{\mathbb{R}^3} A \left(f_1\nabla_{\bf v}f-f\nabla_{{\bf 
  v}_1}f_1\right)d^3{\bf v}_1
\end{equation}
where $A$ is a linear operator from $\mathbb{R}^3$ to $\mathbb{R}^3$ whose corresponding square matrix has the following entries
\begin{equation}\label{eq:matrixA}
a_{ij} = \frac{1}{4|{\bf v}-{\bf v}_1|}\left(\delta_{ij}-\frac{({\bf v}-{\bf v}_1)_i({\bf v}-{\bf v}_1)_j}{|{\bf v}-{\bf 
  v}_1|^2}\right).
\end{equation}

Suppose $f$ is a solution of the Landau equation 
\[
  \frac{df}{dt}=Q_L(f,f)(t,{\bf x},\bv),
\]
and consider a small perturbation $f+\delta f$ around that solution. 

Since the perturbation $\delta f$ is regarded as small with respect to $f$, we simply deduce the equation for $\delta f$ by linearization of the Landau equation for $f$. In particular, we neglect terms of higher order in $\delta f$ inside the Landau collision kernel, which corresponds to the assumption that $Q(\delta f,\delta f)\approx 0$. 
Allowing for a source term $S'$ in order to take into account that the system we are considering is open, we find that the perturbation $\delta f$ is such that 
\begin{equation}\label{eq:landau-linearized}
  \frac{d\delta f}{dt}=L_B(\delta f)(t,{\bf x},\bv)+S_B(\delta f)(t,{\bf x},\bv) + S',
\end{equation}
where, for fixed $f$, $L_B(\delta f)(t,{\bf x},\bv)$ is an operator which acts locally in $\bv$ on $\delta f$, i.e. 
\begin{equation}\label{eqn:landloss}
  L_B(\delta f)(t,{\bf x},\bv)=
\frac{8\pi\Lambda e^4}{m^2} \nabla \cdot\int_{\mathbb{R}^3} A \left(f_1\nabla_{\bf v}\delta f-\delta f\nabla_{{\bf 
  v}_1}f_1\right)d^3{\bf v}_1,
\end{equation}
while $S_B$ is an operator in integral form on $\delta f$
\[
  S_B(\delta f)(t,{\bf x},\bv)=
\frac{8\pi\Lambda e^4}{m^2} \nabla \cdot\int_{\mathbb{R}^3} A \left(\delta f_1\nabla_{\bf v}f-f\nabla_{{\bf 
  v}_1}\delta f_1\right)d^3{\bf v}_1.
\]

\section{Reduction to the continuity equation}
	
In this section we shall see how the linearized Landau approximation of the Boltzmann equation reduces to the continuity equation. 
This result is obtained for large values of $|\bv|$, whenever the  background solution $f\in C^\infty(\mathbb{R}^7)$ 
decays rapidly in the velocity domain, namely whenever $f$ at fixed time and position is a Schwartz function $f(t,{\bf x},\cdot) \in \mathcal{S}(\mathbb{R}^3)$. The same outcome holds in the case of thermal backgrounds described by Maxwell distributions at small temperatures.

\subsection{Asymptotic expansion of the linearized Landau collision kernel}

We are interested in studying the asymptotic behavior of $S_B$ and $L_B$ for large $|\bv|$, whenever the background solution $f$ is of Schwartz type in the velocity domain. This is the case,  for example, whenever $f$ is a Maxwell distribution or a shifted Maxwell distribution.

Let us start discussing $S_B$ in \eqref{eq:landau-linearized}. Since $S_B$ is a local operator in $f$, the asymptotic behavior of $S_B$
for large values of $|\bv|$ is governed by $f$, whenever $\delta f$ is sufficiently regular as is assumed through this paper.
Hence, since $f\in \mathcal{S}(\Omega)$, $S_B$ decays faster than any inverse power of $|\bv|$ for large $|\bv|$.

Let's focus now on the local term $L_B(t,{\bf x},\bv)$ in \eqref{eq:landau-linearized}, where we expect to find the dominant contribution for large values of $|\bv|$.
Integration over all possible target velocities ${\bf v}_1$, recalling the form of the operator $A$ given in \eqref{eq:matrixA}, provides
\begin{eqnarray}
    L_B(\delta f)(t,{\bf x},\bv)&=&-\frac{8\pi\Lambda e^4}{m^2}\Bigg[-2\pi\delta f f-\frac{\Delta\delta f}{4}\int_{\R}\frac{f_1}{|\vmv|}d^3v_1
    \notag
    \\ 
    &&+\frac{\partial_{ij}\delta 
    f}{4}\int_\R\frac{(\vmv)^i(\vmv)^j}{|\vmv|^3}f_1d^3v_1\Bigg].
\label{eqn:lossint}
\end{eqnarray}
In order to discuss the asymptotic form of all these contributions, we need the following theorem   
   \begin{theorem}\label{theo:limite}   
   Let f be a function in $\mathcal{S}(\R)$ and $n:=\int_\R f_1 d^3v_1$. The following relations hold
   \begin{remunerate}
     \item
     $\displaystyle\int_{\R}\displaystyle\frac{f_1}{|\vmv|}d^3v_1=\displaystyle\frac{n}{|\bv|} + \displaystyle R_1(\bv),$
      \item
      $\displaystyle\int_\R\displaystyle\frac{(\vmv)^i(\vmv)^j}{|\vmv|^3}f_1d^3v_1 =\displaystyle\frac{\bv^i\bv^j}{|\bv|^3}n+ \displaystyle R_2(\bv),$   
    \item $\displaystyle\int_{\R}\displaystyle\frac{\partial_i f_1}{|\vmv|}d^3v_1=\displaystyle\frac{\bv^i}{|\bv|^3}n+ \displaystyle R_3(\bv),$
  
   \end{remunerate}
   where $R_1$ and $R_2$ are functions that vanish faster than $1/|\bv|$ for large $|\bv|$, while $R_3$ is a function that vanishes faster than $1/|\bv|^{2}$ for large values of $|\bv|$.\\
   \end{theorem}

The proof of the theorem depends on the 
   \begin{lemma}\label{lemma:limitaz1}
     Let $f(\bv)$ be a function in $\mathcal{S}(\R)$, $\bv\in\R$ and $h(\bv) := |\bv|f(\bv)$. The Fourier transform of the Green operator applied to $h$, namely the Fourier transform of  
     \begin{equation}
      \mathcal{G}(h) :=\int_{\R}\frac{|\bvun|f_1}{|\vmv|}d^3v_1,
     \end{equation}
     is a function in the $L^1$ space.
   \end{lemma}
  {\em Proof of Lemma \ref{lemma:limitaz1}.}
     Since $f$ is
     a Schwartz function, both $h$ and $\Delta h$ are in $L^1(\mathbb{R}^3)$. 
     For the Riemann--Lebesgue lemma their Fourier transform $\mathcal{F}(h)$ and $\mathcal{F}(\Delta h)$ are uniformly bounded. 
     By standard properties of the Fourier transform, we have that 
      \begin{equation}\label{eqn:maggioraz}
        \mathcal{F}(h)(\bk)\leq \frac{C}{(1+|{\bk}|^2)}
     \end{equation}     
      for some positive constant $C$.
     Let's recall now that the Green operator in Fourier space is a multiplicative operator namely 
      $\mathcal{F}(\mathcal{G}(h))(\bk)  = \mathcal{F}(h)(\bk) {|\bk|^{-2}}$, hence 
      \[
      |\mathcal{F}(\mathcal{G}(h))(\bk)| \leq \frac{C}{|\bk|^2(1+|{\bk}|^2)}
      \]
 	 which means that $|\mathcal{F}(\mathcal{G}(h))(\bk)|$ is bounded by an element of $L^1(\mathbb{R}^3)$. So, we have the thesis.\qquad\endproof

   We can now proceed with the proof of Theorem \ref{theo:limite}.
    {\em Proof of Theorem \ref{theo:limite}}
      We will first prove $1)$ computing the limit 
\[
\lim_{|\bv|\to\infty} |\bv|\int_{\mathbb{R}^3}\frac{f_1}{|\vmv|}d^3v_1 =
\lim_{|\bv|\to\infty}  \left[ \int_{\R}\frac{|\bv|-|\bvun|}{|\vmv|}f_1d^3v_1+\int_{\R}\frac{|\bvun|f_1}{|\vmv|}d^3v_1 \right]
\]
where we used the linearity of the Green operator in order to divide the desired limit in two parts. 
Let us discuss them separately. 
As shown in Lemma \ref{lemma:limitaz1}, the second integral in the previous expression is the inverse Fourier transform of an $L^1$ function, hence, by Riemann-Lebesgue lemma, its limit for large $|\bv|$ vanishes. 
Since the absolute value of the function ${(|\bv|-|\bvun|)}/{|\vmv|}$ is bounded by $1$ uniformly in $\bv$, 
the limit of the first integral can be obtained applying the dominated convergence theorem, namely taking the limit before computing the integral. 
The final result is 
\[
\lim_{|\bv|\to\infty} |\bv|\int_{\mathbb{R}^3}\frac{f_1}{|\vmv|}d^3v_1  = n
\]
and this proves the first assertion.

Let's proceed with $2)$. We observe that
  \begin{gather*}
|\bv| \int_\R \frac{(\vmv)^i(\vmv)^j}{|\vmv|^3}f_1d^3v_1=
\int_\R\frac{(|\bv|-|\bvun|)}{|\vmv|}\frac{(\vmv)^i(\vmv)^j}{|\vmv|^2}f_1d^3v_1+\\
+\int_\R\frac{|\bvun|}{|\vmv|}\frac{(\vmv)^i(\vmv)^j}{|\vmv|^2}f_1d^3v_1.
\end{gather*}
 
For large values of $|\bv|$, the limit of the first term of the last sum can be obtained applying 
the dominated convergence theorem.
The second term tends to zero as $|\bv|$ is very large. In fact, as $f_1$ is positive,
\begin{equation*}
 0\leq\left|\int_{\R}\frac{|\bvun|}{|\vmv|}\frac{(\vmv)^i(\vmv)^j}{|\vmv|^2}f_1d^3v_1\right|\leq 
  \int_{\R}\frac{|\bvun|f_1}{|\vmv|}d^3v_1.
  \end{equation*}
We can thus apply an argument similar to the one employed in the proof of $1)$ to show that it vanishes.
Hence we have that 
\[
\lim_{|\bv|\to\infty}|\bv| \int_\R \frac{(\vmv)^i(\vmv)^j}{|\vmv|^3}f_1d^3v_1= \frac{\bv^i\bv^j}{|\bv|^3}n,
\]
from which follows the thesis.

    In order to prove $3)$ we observe that
	\[
      |\bv|\int_\R\frac{\partial_if_1}{|\vmv|}d^3v_1 =  \int_\R\frac{(|\bv|-|\bvun|)}{|\vmv|}\partial_i 
      f_1d^3v_1+\int_\R\frac{|\bvun|}{|\vmv|}\partial_i 
      f_1d^3v_1.
   	\]
    Since $f\in\mathcal{S}(\R)$,
    also its partial derivatives are in $\mathcal{S}(\R)$, i.e. $\partial_i f\in\mathcal{S}(\R)$ 
    for each $i=1,2,3$, and we can thus proceed as in in the proof of $1)$. \qquad\endproof

Applying the result of the previous theorem to \eqref{eqn:lossint} and recalling the discussion about the decaying properties of $S_B$ given at the beginning of this section, we have the following

\begin{proposition}\label{pr:asymptotic-QL}
The linearized Landau integral kernel $Q_L = L_B + S_B$, computed over a background distribution $f\in C^\infty$ which is of Schwartz type in the velocity domain, is such that 
\begin{equation*}
  Q_L(\delta f) = L_B(t,{\bf x},\bv)+ S_B(t,{\bf x},\bv) = -\frac{8\pi\Lambda e^4}{m^2}\left(-\frac{n}{4|\bv|}\Delta\delta f+\frac{n}{4}\frac{\bv^i\bv^j}{|\bv|^3}\partial_{ij}\delta f\right) + R.
\end{equation*}
where $R$ is function which vanishes faster than $1/|\bv|$ for large values of $|\bv|$.
\end{proposition}

We specialize now the asymptotic expansion provided by the previous proposition for backgrounds described by Maxwell distribution
\begin{equation}\label{eq:maxwell-distribution}
f(t,{\bf x},{\bf v}) = \frac{n (m\beta)^{\frac{3}{2}}}{\sqrt{8}\pi^{\frac{3}{2}}}e^{-\frac{mv^2\beta}{2}},
\end{equation}
where $\beta={(k_bT)}^{-1}$, $k_b$ is the Boltzmann constant and $T$ the temperature of background particles. 
In particular, we have that for large $\beta$, i.e. in the limit of vanishing temperature (cold target), the asymptotic expression found above becomes exact, and so the following proposition holds.

\begin{proposition}\label{pr:asymptotic-cold}
The linearized Landau kernel $Q_L=L_B(t,{\bf x},\bv) + S_B(t,{\bf x},\bv)$ defined as in \eqref{eq:landau-linearized} for perturbations $\delta f$ over a Maxwell distribution $f$ as in \eqref{eq:maxwell-distribution} at inverse temperature $\beta$ is such that 
\[
\lim_{\beta\to\infty} Q_L(\delta f)(t,{\bf x},\bv)
= -\frac{8\pi\Lambda e^4}{m^2}\left(-\frac{n}{4|\bv|}\Delta\delta f+\frac{n}{4}\frac{\bv^i\bv^j}{|\bv|^3}\partial_{ij}\delta f\right)\qquad |\bv|\neq 0.
\] 
\end{proposition}

\begin{proof}
In order to prove this statement, it is sufficient to notice that, in the limit of large $\beta$, the Maxwell distribution
is proportional to a regularization of the Dirac delta centered in zero, where the factor of proportionality is noting but the plasma density.  
Hence, being $S_B$ local in $f$, it vanishes for $|\bv|\neq 0$. At the same time, analyzing $L_B$ in \eqref{eqn:lossint} we have the thesis.\qquad
\end{proof}

In the previous proposition  we have seen that the leading contribution in the linearized Landau kernel $Q_L$ for large values of $|\bv|$, 
is expressed by the following operator
\begin{equation*}
  \tilde{L}_B(t,{\bf x},\bv)=
  \frac{8\pi\Lambda e^4}{m^2}\left(\frac{n}{4|\bv|}\Delta\delta f-\frac{n}{4}\frac{\bv^i\bv^j}{|\bv|^3}\partial_{ij}\delta 
f\right).
\end{equation*}
We shall now specialize the form of $\tilde{L}_B$ for $\delta f$ which are isotropic in $\bv$ namely which are functions of $|\bv|$,
\begin{equation}\label{eq:isotropic}
\delta f(\bv) = g (|\bv|).
\end{equation}
For this kind of functions, by direct computation, it holds that 

\begin{equation}\label{eq:asymptotic-landau-loss-term}
  \tilde{L}_B(t,{\bf x},\bv)=
  \frac{8\pi\Lambda e^4}{m^2}\frac{n}{2|\bv|^2} g'   =  
  \frac{8\pi\Lambda e^4}{m^2}\frac{n}{2|\bv|^3} \bv\cdot\nabla \delta f.
\end{equation}

\subsection{Asymptotic form of the energy loss term and its relation with the Landau integral kernel}

In the previous section we have analyzed the asymptotic form of the loss term \eqref{eqn:landloss} of the Landau equation \eqref{eq:landau-linearized} in the limit of high values of electron velocities, which is equivalent to the limit of small temperatures in the case of thermal background. 
Here we compare the asymptotic form of the Landau collision kernel obtained so far with a stochastic loss term added to a continuity equation \eqref{eq:continuity-stochastic} for the linearized perturbations $\delta f$.
Let us recall the equation for the linearized perturbation $\delta f$
\[ \frac{\partial \delta f}{\partial t} +
\bv \cdot  \nabla_{\bf x}  \delta f+
\frac{{\bf F}}{m}  \cdot \nabla_\bv   \delta f 
 = -\frac{\langle\nabla_{\bf x} E\rangle}{m}\cdot\nabla_{\bf v}\delta 
   f+S(t,{\bf x},{\bf v}),
\]

 where $\langle\nabla_{\bf x} E\rangle$ is the energy loss due to collisions. The loss term due to collisions in the previous equation is thus given by
\begin{equation}\label{eq:LC}
 L_C(\delta f)(t,{\bf x},\bv):=  -\frac{\langle\nabla_{\bf x} E\rangle}{m}\cdot\nabla_{\bf v}\delta f.
\end{equation}
Let us now briefly discuss the form of the stochastic force $\langle \nabla_{{\bf x}}E \rangle$ due to Coulomb collision of the perturbed particles distributed as  $\delta f$ with background particles in the ensemble $f$.
We start recalling that the Coulomb collisional energy loss rate obtained in \cite{sp62,bubu62,lg63,em78} is
 \begin{equation*}
   \left\langle \frac{dE}{dt} \right\rangle=\int_{\mathbb{R}^3}\int_{S^2}B(|{\bf v}-{\bf v}_1|,\sigma)\Delta 
   T f_1 d^2\sigma d^3v_1
 \end{equation*}
 where $B$ is defined in terms of the Rutherford cross section as in \eqref{eqn:Ruth} and  $\Delta T$ is the variation of the kinetic energy of an accelerated electron due to a single collision, namely
\[
 \Delta T=-\frac{m|\vmv|}{2}{\omega}\cdot\bv-\frac{m|\vmv|^2}{4}(\cos\theta-1),
\] 
where ${\omega}= \sigma -(0,0,1) = (\sin\theta\cos\phi,\sin\theta\sin\phi,\cos\theta-1)$ is the deviation angle \cite{bubu62}.
 Performing the integral over all possible scattering angles and
recalling that $\left\langle \frac{dE}{dt} \right\rangle={\bf v}\cdot\langle \nabla_{{\bf x}}E \rangle$, 
the stochastic force results as being
 \begin{equation*}
\langle \nabla_{{\bf x}}E \rangle=-\frac{8\pi\Lambda e^4}{m}\left[\int_{\mathbb{R}^3}\frac{\nabla_{{\bf v}_1}f_1}{|{\bf v}-{\bf v}_1|}d^3v_1-\frac{\bf v}{2|{\bf v}|^2}\int_\R \frac{f_1}{|\vmv|}d^3v_1 
   \right].
 \end{equation*}
From this expression it is easy to write the loss term $L_C$ and, furthermore, its asymptotic form can be obtained using Theorem \ref{theo:limite}.
In particular the following proposition holds
 \begin{proposition}
Let $f$ be any smooth function which is of Schwartz type in the velocity domain. The asymptotic expansion of $L_C(t,{\bf x},\bv)$ defined in \eqref{eq:LC} for large $|\bv|$  
 is such that
 \begin{equation*}
   L_C(t,{\bf x},\bv) =  \tilde{L}_C(t,{\bf x},\bv) +R = \frac{8\pi\Lambda e^4}{m^2}\frac{n}{2|\bv|^3}\bv\cdot\nabla_{\bv}\delta f+ R   
 \end{equation*}

 where, for large values of $|\bv|$, $R$ is a function that vanishes faster than $1/|\bv|^3$. \\ 
\end{proposition}
When the background is thermal, namely when $f$ is a Maxwell distribution as in \eqref{eq:maxwell-distribution}, proceeding as in the proof of proposition \ref{pr:asymptotic-cold} we have that the asymptotic form obtained in the previous proposition implies that the limit of vanishing background temperature is such that
\[
\lim_{\beta\to \infty} L_C=\tilde{L}_C=\frac{8\pi\Lambda e^4}{m^2}\frac{n}{2|\bv|^3}\bv\cdot\nabla_{\bv}\delta f,\qquad  \bv\neq 0.
\]
Finally, by direct inspection, we notice that whenever $\delta f$ is an isotropic perturbation as in \eqref{eq:isotropic}, the asymptotic form of the stochastic loss term of the continuity equation $\tilde{L}_C(t,{\bf x},\bv)$ matches the asymptotic form of the Boltzmann collision kernel $\tilde{L}_B(t,{\bf x},\bv)$ given in \eqref{eq:asymptotic-landau-loss-term}.
This observation suggests that the correct form of the modified continuity equation is \eqref{eq:continuity-stochastic} and not \eqref{eq:continuity-stochastic-wrong} as sometimes assumed in the literature.
In particular, the stochastic force due to Coulomb collisions needs to be taken into account adding to the continuity equation a contribution  of the form 
\begin{equation}\label{eqn:continuity-loss}
L_C(\delta f)(t,{\bf x},\bv)=\frac{\langle\nabla_{\bf x} E\rangle}{m} \cdot \nabla_{\bf v}\delta f.
\end{equation}

\section{Resolution of the continuity equation with hyperbolic methods} 

The aim of this section is to write the continuity equation and its solution in a form suitable to be compared with data obtained by 
satellite imaging-spectrometer observing solar flares.
At this point it is important to recall that these instruments 
measure {\em photon fluxes} resolved in space, time and energy
However, since photons are produced by accelerated electrons due to Bremsstrahlung, it is in principle possible to directly obtain information about electron ensemble from the measured photon fluxes.
This has been accomplished by Piana et al. in \cite{pmh07}  where they employed refined regularization techniques to invert the Bremsstrahlung cross section.
The outcomes  of the analysis presented in \cite{pmh07} are 
maps of {\em electron fluxes} $F(t,{\bf x},E)$, which are noting but images resolved in space and (kinetic) energy. 

It is at this point interesting to compare theoretical predictions for electron densities $f$, following from certain kinetic equations like \eqref{eq:continuity-stochastic}, with observations of electron fluxes $F$ obtained from satellite measurements. 

However, since electron fluxes are integrated quantities over all velocity directions, this comparison is not completely straightforward.

For this reason, starting from the precise relation among $F$ and $f$, we transform \eqref{eq:continuity-stochastic} into an equation valid for the electron flux $F(t,{\bf x},E)$. Afterwards  we shall solve this equation using some hyperbolic techniques to obtain a solution comparable with measurements. 

Let us recall that, at fixed time, $F$ describes the flux density over the domain $\mathbb{R}^3\times \mathbb{R}^+$ parametrized by $({\bf x},E)$, while $f$ describes the particle density over $\mathbb{R}^6$ parametrized by $({\bf x},\bv)$. 
Since $E$ is the kinetic energy of a particle of velocity $\bv$, the definition of $F$ in terms of $f$ is
\begin{equation}\label{eq:velocity-angles}
F(t,{\bf x},E) :=   \frac{|\bv|^2}{m}\int_{{\mathbb S}^2} f(t,{\bf x},\bv) d{\mathbb S}^2  \;,\qquad E = \frac{1}{2}m|\bv|^2,
\end{equation}

where the integration is over all possible velocity directions and hence 
$d{\mathbb S}^2$ is the standard measure of the two-dimensional sphere $\mathbb{S}^2$.
Similarly, having in mind \eqref{eq:continuity-stochastic}, the same averaging needs to be applied on the source term. Then, we can define the \emph{injection flux} $\Sigma$ related to the source $S$ to be
\begin{equation}\label{eq:source-vel-ang}
\Sigma(t,{\bf x},E) := \frac{|\bv|^2}{m} \int_{{\mathbb S}^2} S(t,{\bf x},{\bf v}) d{\mathbb S}^2.
\end{equation}
Although the information about the velocity directions is lost in $F$, 
under some particular assumptions we are going to discuss, we might be able to extract an equation for electron fluxes $F$ from \eqref{eq:continuity-stochastic}.

First of all,  in the limit of small background temperature, the leading contribution in the electron distribution $f(t,{\bf x},\bv)$ is given by accelerated electrons only. For this reason, the thermal background might be discarded and thus $f$ obeys the continuity equation equipped with the stochastic force \eqref{eq:continuity-stochastic}, 
\begin{equation}\label{eq:continuity}
 \frac{\partial f}{\partial t} +
  {{\bf v}} \cdot\nabla_{\bf x} f=-\frac{1}{m}\langle\nabla_{\bf x} E\rangle \cdot \nabla_{\bf v} f+S(t,{\bf x},{\bf v})
\end{equation}
where $S(t,{\bf x},{\bf v})$ is some external source and  $\langle\nabla_{\bf x} E\rangle$ is the energy loss rate.

Magnetic fields are particularly strong during solar flares. We thus assume that electrons move freely along the magnetic field lines discarding the rotational motion in perpendicular planes. At the same time we expect that the electron distribution $f$ does not vary much passing to nearby magnetic field lines.
Under this assumption it is convenient to parametrize the space with an appropriate coordinate system $(s,x,y)$, and hence 
$s$ is the arc--length of a magnetic field line $\gamma$ while $\partial/\partial x$ and $\partial/\partial y$ are orthogonal to $\gamma$.
 Thanks to this choice $f(t,{\bf x},\bv)=f(t,s,\bv)$, namely the distribution $f$ depends on space only through $s$ while it is constant under changes of $x,y$ coordinates.  
In order to obtain an equation for the flux density $F$, we assume that both $f$ and the source $S$ in \eqref{eq:velocity-angles} and in  \eqref{eq:source-vel-ang} depend on the velocities only through $|\bv|$, hence $f=f(t,s,|\bv|)$. 

We proceed writing the accelerated electron ensemble as the sum of two parts 
$f=f_++f_-$, 
where particles in $f_+$ move towards positive coordinates $s$ while those in $f_-$ towards negative ones. In other words, $f_+$ (resp. $f_-$) vanishes when $v_s$, the component of the velocity $\bv$ along $\gamma$, is negative (resp. positive). Furthermore,  \eqref{eq:continuity} holds separately for $f_+$ and $f_-$, 
hence
\begin{equation}\label{eq:temp_cont}
 \frac{\partial f_\pm}{\partial t} +
  v_s \frac{\partial f_\pm}{\partial s}=-\frac{1}{m}\langle\nabla_{\bf x} E\rangle \cdot \nabla_{\bf v} f_\pm+S(t,{\bf x},|{\bf v}|).
\end{equation}
To obtain an equation for the flux densities  $F_\pm$ corresponding to $f_\pm$, we proceed by computing the integral of the previous two equations on the velocity directions as in \eqref{eq:velocity-angles} and in \eqref{eq:source-vel-ang}.
Thanks to the isotropic hypothesis for $f$ and $S$, the right hand side of \eqref{eq:temp_cont} does not depend on the velocity directions, while the left hand side depends on the directions only through $v_s$. Since the average value of $v_s$ over the velocity directions, is $|\bv|/2$ (resp. $-|\bv|/2$) for particles in $f_+$ (reps. $f_-$), the integrals gives
\begin{equation}\label{eq:continuity-1}
\frac{\partial}{\partial t} \frac{F_\pm(t,s,E)}{E}  \pm  \sqrt{\frac{E}{2m}} \frac{\partial}{\partial s} \frac{F_\pm(t,s,E)}{E} +\left\langle\frac{dE}{ds}\right\rangle\sqrt{\frac{2E}{m}}\frac{\partial}{\partial E}\left(\frac{F_{\pm}(t,s,E)}{E}\right)  = \frac{\Sigma_{\pm}(t,s,E)}{E},
\end{equation}
where $\Sigma_+(t,s,E)$  (resp. $\Sigma_-$)  are obtained by an integration of $S(t,s,\bv)$ as in \eqref{eq:velocity-angles} over velocity directions corresponding to positive (resp. negative) $v_s$. Notice that, since $S$ is assumed to be isotropic in $\bv$, we have that $\Sigma_+=\Sigma_-$.

We shall now look for static solutions of the previous equation. Assuming $\Sigma_+=\Sigma_-=\Sigma(s,E)$ to be constant in time, \eqref{eq:continuity-1} reads
\begin{equation}\label{eq:continuity-2}
\pm   \frac{\partial}{\partial s} \frac{F_\pm(s,E)}{E}   +2\left\langle\frac{dE}{ds}\right\rangle\frac{\partial}{\partial E}    \frac{F_\pm(s,E)}{E}   = \sqrt{\frac{2m}{E}}\frac{\Sigma (s,E)}{E}.  \;
\end{equation}

We shall now  simplify the problem by means of the following substitutions\footnote{The definition of $x$ is obtained imposing $\frac{dE}{dx}=2\left\langle\frac{dE}{ds}\right\rangle$ and supposing $x(\bar{E})=\bar{x}$}:
\begin{equation}\label{substitution}
A(s,E)=\sqrt{\frac{2m}{E}}\frac{\Sigma(s,E)}{E}   \;, \qquad  x(E)= \bar{x}+\int_{\bar{E}}^E \frac{1}{2\left\langle\frac{dE}{ds}\right\rangle}dE'\;,
\qquad
\Phi_\pm = \frac{F_{\pm}}{E}.
\end{equation}

Then \eqref{eq:continuity-2} take the form 
\begin{equation}\label{eq:temp}
\pm   \frac{\partial}{\partial s}   \Phi_\pm  -    \frac{\partial}{\partial x} \Phi_\pm     = A  \;.
\end{equation}
Since a constant solution is not physically admissible neither for $F_+$ nor for $F_-$,  
we can apply on both sides of \eqref{eq:temp} the operators
$$
\pm   \frac{\partial}{\partial s}   +    \frac{\partial}{\partial x}\;
$$
without loosing any information.
The mixed terms disappear and thus both equations for $\Phi_+$ and $\Phi_-$ can be combined to give 
\begin{equation}\label{eqn:dalemb}
\Box\Phi       = 2 \frac{\partial}{\partial x} A  \;,
\end{equation}
where $\Phi= \Phi_++\Phi_- = F/E$ is the sum of forward and backward moving electrons and $\Box$ is the D'Alembert operator in two dimensions. 
Equation \eqref{eqn:dalemb} is nothing but the wave equation in two dimensions with a source term, where $x$ plays the role of time and $s$ of space.

To find the solution of this equation we have to provide physically motivated initial conditions.
Since electrons cannot gain energy during their motion (injected electrons can only loose energy during their path), 
we shall look for solutions which satisfy the following initial conditions
\begin{equation}\label{inicond}
\lim_{x\to \infty} \Phi(x,s\pm x)=0\;,\qquad \lim_{x\to \infty} \partial_x \Phi(x,s\pm x)=0\;,\qquad \forall s\;.
\end{equation}
The solution can now be obtained by means of the convolution (in the $x,s$ space) of the advanced (with respect to $x$) fundamental solution (Green function) $\delta_{-}=\delta{(x+|s|)}$ of the D'Alembert equation with the source term $2A$, namely
$$
\Phi = \delta_- * 2A \;.
$$
The absence of the derivative $\partial_x$ in front of $A$ in the previous expression, contrary to
the case of  \eqref{eqn:dalemb}, is due to the fact that $\Box\delta_- = \partial_x \delta $. 
Actually, in two dimensions such a derivative is necessary in order to avoid infrared divergences.
Furthermore, $F$ obtained in that way is the unique solution obeying the desired initial conditions 
\eqref{inicond}.
All other solutions of the wave equation with the chosen source
do not vanish for large energies.

We can now perform the convolution using \eqref{substitution}, namely
\begin{equation}\label{eq:continuity-solution-noint}
F(s,E) =  2\,E  \int  \delta\left( x - x_0 +  |s-s_0|   \right) \sqrt{\frac{2m}{E(x_0)}}\frac{\Sigma (s_0,E(x_0) )}{E(x_0)}
 dx_0  ds_0,
\end{equation}
where $x(E)$ is defined as in \eqref{substitution} and $E(x_0)$ is the inverse function of $x(E)$ evaluated in $x_0$.
So, \eqref{eq:continuity-solution-noint} gives 
\begin{equation}\label{eq:continuity-solution}
F(s,E) = 2E \int\sqrt{\frac{2m}{E(x+|s-s_0|)}}\frac{\Sigma (s_0,E(x+|s-s_0|) )}{E(x+|s-s_0|)}ds_0\;.
\end{equation}
Notice that the solution written in the form \eqref{eq:continuity-solution}, is in accordance with the one presented in \cite{guo13} under the assumption of uniform injection in a limited region of the flare. Furthermore, this expression can be directly compared with results provided by experiments, as in \cite{torre, codi13}.\\

More precisely, as in the case of cold target \cite{em78}, the energy loss rate is obtained from 
\[
\left\langle \frac{dE }{ dt } \right\rangle = -2\pi e^4\Lambda n\sqrt\frac{2}{m E},  
\] 
letting $K=2\pi\,e^4\,\Lambda$, we have 
\[
x(E)=\frac{E^2}{4Kn}\qquad \mbox{and hence}\qquad E(x)=2\sqrt{Kn x},
\]
and \eqref{eq:continuity-solution} takes then the form
\begin{equation}\label{eq:continuity-solution-cold}
F(s,E) = 2E \int\sqrt{\frac{2m}{\left(  { E^2}+  4nK|s-s_0|   \right)^{1/2}}}\frac{ \Sigma(s_0,\left(  { E^2}+  4nK|s-s_0|   \right)^{1/2} ) }{\left(  { E^2}+  4nK|s-s_0|   \right)^{1/2}}ds_0\;.
\end{equation}

\section{Comparison with experiments}

In this section we validate the solution of the continuity equation obtained so far for accelerated electron fluxes with the results of a real measurement occurred during a solar flare. 
First of all, we compare the outcome of our analysis with the continuity equation for the electron flux described by Emslie, Barrett and Brown \cite{emsbarbro01}. Following the notation of the previous sections this equation reads
\begin{equation}\label{eq:class-cont}
	\frac{\partial F(s,E)}{\partial s}+\frac{\partial}{\partial E}\left(\left\langle\frac{dE}{ds}\right\rangle F(s,E)\right)=\sqrt{\frac{2m}{E}}\Sigma(s,E).
\end{equation}
Applying the same procedure used to solve \eqref{eq:continuity-1}, we obtain a solution for \eqref{eq:class-cont}, in the form of
\begin{equation}\label{eq:class-continuity-solution}
F(s,E)=2 \left\langle \frac{dE}{ds} \right\rangle^{-1} \int\left\langle \frac{dE}{ds} (E(x'))\right\rangle \sqrt{\frac{2m}{E(x')}}\Sigma(s_0, E(x'))ds_0,
\end{equation}
where $x' = x+|s-s_0|$ and, in contrast to \eqref{substitution}, $x$ is defined as
\begin{equation}\label{eq:x-class}
x(E)= \bar{x}+\int_{\bar{E}}^E \frac{1}{\left\langle\frac{dE}{ds}\right\rangle}dE'.
\end{equation}
From a theoretical point of view, the expressions of $F(s,E)$ given by \eqref{eq:continuity-solution} and \eqref{eq:class-continuity-solution} differ by two aspects. 
Firstly, the energy loss rate enters  
\eqref{eq:continuity-solution} and \eqref{eq:class-continuity-solution} in two different ways. The resulting electron fluxes are compared in Figure \ref{fig:solcfr13} where we notice a significant difference at lower energies, e.g. 1-2 keV.
Secondly, \eqref{eq:continuity-1} is deduced from \eqref{eq:continuity} performing the integral over all particle directions.
In this way, collisions along every direction are taken into account. On the contrary, in the derivation of  \eqref{eq:class-cont} only interactions along the magnetic field lines are considered.
These different assumptions have an explicit effect on the expression of $x(E)$, used in \eqref{substitution} and in \eqref{eq:x-class}, making them differ by a factor two.
This modification is evident also at higher energies, in particular at $10-100$ keV, the energy range in which measurements can be detected by imaging-spectrometers. \\
\begin{center}
\begin{figure}[!ht]
\includegraphics[width=.5\textwidth, angle=90]{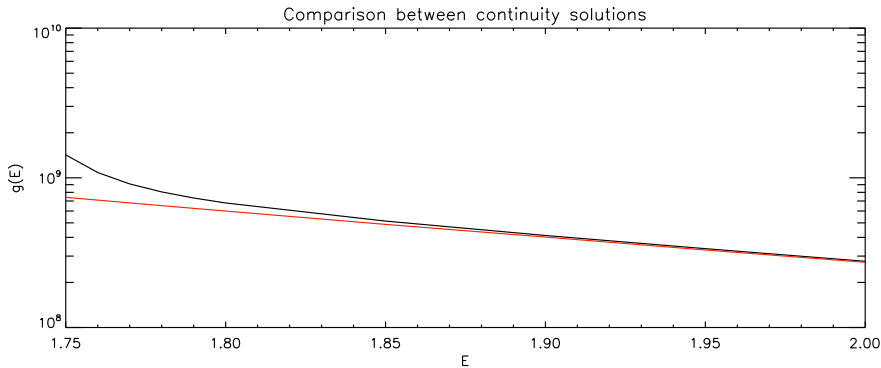}
\caption{In red and black the profiles of $g(E)$ obtained respectively with $F(s,E)$ measured as in \eqref{eq:continuity-solution} and in \eqref{eq:class-continuity-solution} in the electron energy range between 1,75 keV and 2 keV.}\label{fig:solcfr13}
\end{figure}
\end{center}
We check now that the electron flux distributions described in \eqref{eq:continuity-solution} and in \eqref{eq:class-continuity-solution} are both in accordance with the available experimental data. To accomplish this task 
we recall that images of electron fluxes deduced by the NASA Reuven Ramaty High Energy Solar Spectroscopic Imager (RHESSI) spectrometer  \cite{ldh02} can be analyzed. We shall thus validate both solutions \eqref{eq:continuity-solution} and \eqref{eq:class-continuity-solution} with the observations of a single event.
Furthermore, regarding the form of the source considered in the theoretical solutions, 
we assume that the injected electrons are constant along a certain region  
of the flare with length $L_0$
 i.e.
\begin{equation}\label{eq:source_box}
\Sigma(s,E)=h_s\sqrt{\frac{E}{2m}} \left(\frac{E}{E_0}\right)^{-\delta}\Theta\left(\frac{L_0}{2}-|s-s_m|\right),
\end{equation}
where $\Theta$ is the Heaviside function, $h_s$ is a constant which represents the intensity of the injection, $E_0$ is the minimum injection energy and $s_m$ represents the origin of the arc-length system of coordinates. Moreover, the hot target model described in \cite{sp62} is used as energy loss rate, i.e.,
\begin{equation}\label{eq:loss_hot}
\frac{dE}{ds}=-\frac{Kn}{E}\left[\text{erf}\left(\sqrt{\frac{E}{kT}}\right)-\frac{4}{\sqrt{\pi}}\sqrt{\frac{E}{kT}}e^{-\frac{E}{kT}}\right],
\end{equation}
where $n$ represents the background electron density, $kT$ is the plasma temperature measured in keV and $\text{erf}$ is the ordinary error function.
As experimental data, we consider the RHESSI acquisitions of a flare, which occurred on the 15th of April 2002 in the time range between 00:00 and 00:05.\\ 

More in details, we have selected an electron energy range between 18 keV and 34 keV. We have separated this interval into eight sub-intervals of amplitude of 2 keV and we have recovered the correspondent electron flux. In particular, using the method described in \cite{pmh07}, we have obtained eight electron flux maps similar to the one shown in the first column of Figure \ref{fig:arc}. 

As discussed in the theoretical analysis presented in \S 4, we have to restrict the attention to the motion of accelerated electrons along a single fixed magnetic field line. The magnetic field line, or better the electron path, we have considered has thus been selected as the locus of the brightest points of each image row, following the procedure described in \cite{torre}. The value of each pixel along this path represents the measured electron flux $F(s,E)$.
The chosen pixels and the relative electron flux values are shown in the second and the third columns of Figure \ref{fig:arc} respectively.  

Then, for each energy channel, we have computed the mean value of the electron flux along the path, obtaining an averaged electron flux $g(E)$ of the form
\begin{equation}
	g(E)=\int F(s,E) ds.
\end{equation}
The measured values of $g(E)$ are shown in the first row of Figure \ref{fig:fit_15apr}.

We have thus deduced this same quantity using the theoretical models described in \eqref{eq:continuity-solution} and \eqref{eq:class-continuity-solution} for the electron flux, provided a source term and an energy loss rate as described in \eqref{eq:source_box} and \eqref{eq:loss_hot}. 
Hence, the measured values of $g(E)$ have been fitted against theoretical ones using a simulated annealing algorithm \cite{simulatedannealing} with the mean squared error percentage as an energy function.
The best values for $h_s$ and $\delta$ have been obtained, while the target electrons density, the plasma temperature, the minimum injection energy and the injection amplitude have been fixed as $n=10^{11}\,\text{cm}^{-3}$, $kT=1,75\,\text{keV}$, $E_0=18\,\text{keV}$ and $L_0=30\,\text{arcsec}$ respectively.
\begin{figure}[!ht]
\includegraphics[width=.33\textwidth, angle=90]{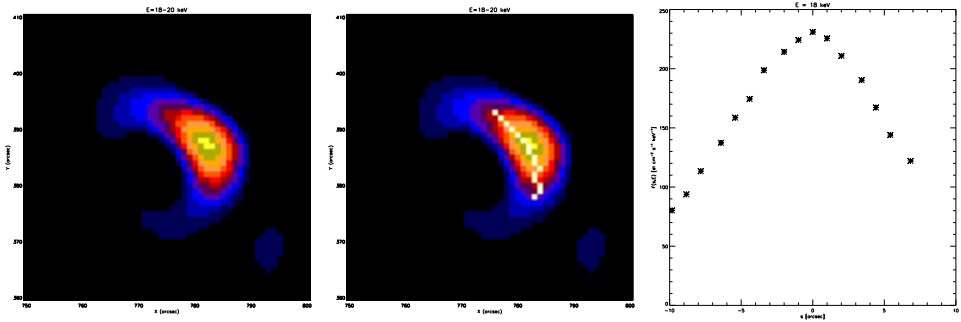}
\caption{In the first column a reconstruction of an electron flux map obtained by the algorithm described in \cite{pmh07} starting from RHESSI data is provided. The accelerated electron path selected following the method described in \cite{torre} is highlighted in white in the second column. The associated pixel values (i.e. the electron flux amounts) are shown in the third column.}\label{fig:arc}
\end{figure}

The values obtained are shown in table \ref{tab:results} while in Figure \ref{fig:fit_15apr} the experimental profile and the fitted solutions are compared. It's interesting to notice that both \eqref{eq:class-continuity-solution} and \eqref{eq:continuity-solution} fit very well the measured data, as both models give a really low value for the mean squared error percentage. Furthermore, the parameters obtained are physically consistent and, actually, they are all similar to those found in \cite{guo13,codi13,torre}. 
On the other hand, although the fitted curves are almost indistinguishable, we stress that parameters obtained using the two models are slightly different. In principle it is thus possible to distinguish between \eqref{eq:continuity-solution} and \eqref{eq:class-continuity-solution}. These difference  should become more evident when more refined measurements will be available.\\

Summarizing, we have compared two different models for the theoretical description of the electron flux, described by \eqref{eq:class-continuity-solution} and \eqref{eq:continuity-solution} respectively. It is important to point out that while the former model has been deduced under an hypothesis of flux conservation, the latter has been developed starting from basic principles and it is still in accordance with experimental data. 

\begin{center}
\begin{figure}[!ht]
\includegraphics[width=.5\textwidth, angle=90]{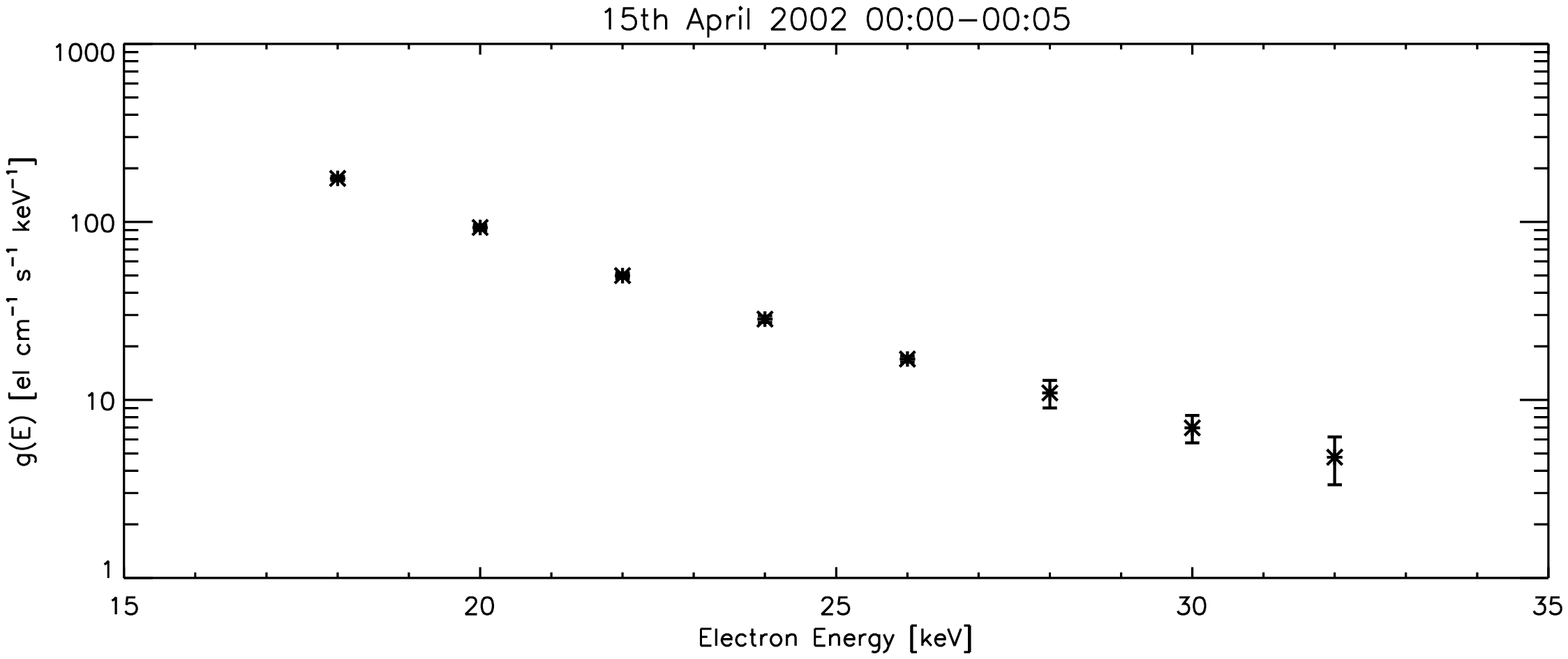}
\includegraphics[width=.5\textwidth, angle=90]{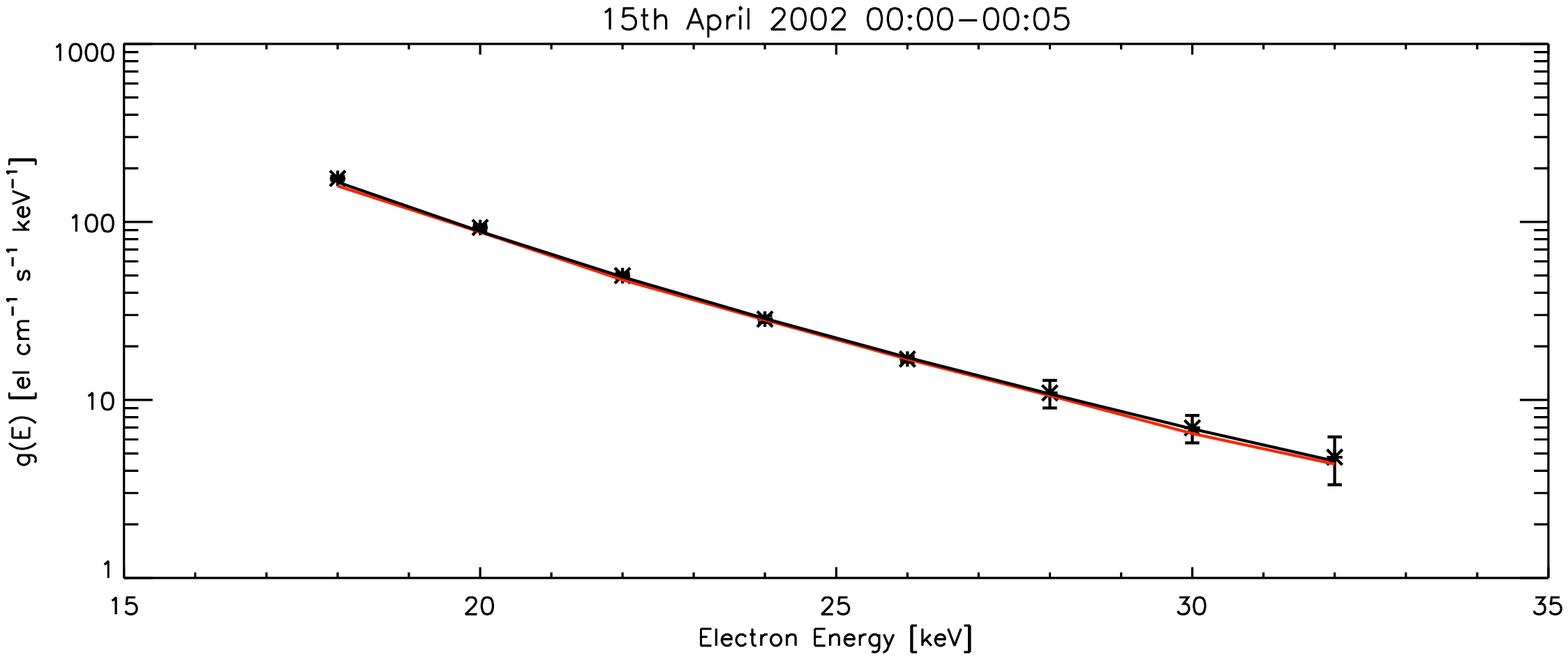}
\caption{2002 April 15 event: $g(E)$ as a function of electron energy. First row: Empirical values of the averaged electron flux. Second row: Empirical values fitted against theoretical models, in red the profile obtained with \eqref{eq:continuity-solution} and in black the profile obtained with \eqref{eq:class-continuity-solution}.}\label{fig:fit_15apr}
\end{figure}
\begin{table}[!ht]
\centering
\footnotesize
\begin{tabular}{|c|c|c|}
\hline
$h_s\,[\text{cm}^{-3}\,\text{s}^{-1}\,\text{keV}^{-1}]$&$\delta$&MSEP\\
\hline
$1.39\times 10^{9}$&$8.12$&$0.0003$\\
$6.81\times 10^{8}$&$7.88$&$0.0006$\\
\hline
\end{tabular}
\caption{Plasma parameters obtained with \eqref{eq:continuity-solution} and \eqref{eq:class-continuity-solution} respectively in the first and in the second row: in the first column the obtained injection intensity, in the second the amplitude of the injection region, in the third the injection spectral index and in the fourth the resulting mean squared error percentage.}\label{tab:results}
\end{table}

\end{center}

\newpage
\section{Comments and open issues}
In this paper we showed that in some asymptotic limits the linearized kinetic Boltzmann equation for Coulomb collisions reduces to a continuity equation equipped with an external stochastic force. 
Specifically, under the assumption of grazing collisions (namely when the Coulomb logarithm is very large), we showed that, in the limit of high values for the electron velocities, the Landau equation reduces to the continuity equation, provided a stochastic loss term of the form 
\[
\frac{\langle \nabla E_{\bf x}\rangle}{m}\cdot \nabla f,
\]
is taken into account. The resulting equation has the form  
\[
\frac{\partial  f}{\partial t} + 
\bv \cdot \nabla_{\bf x}   f+ 
\frac{{\bf F}}{m}\cdot \nabla_\bv f +
\frac{\langle\nabla_{\bf x} E\rangle}{m}\cdot \nabla_\bv f
=  S\;.
\]
We point out that the stochastic contribution written above, is slightly different than the usual loss term considered in some literature about electron motion in solar flares, especially at lower values for the electron energy $E$. 
Since measurements in these domain start becoming very precise this difference should  be appreciated in the close future.
For this reason, in the second part of the paper, 
we have rewritten the continuity equation in an equivalent form ready to be compared with data collected by spectrometer measuring X--ray photons emitted during solar flares. 
Indeed, the continuity equation has been reduced to an hyperbolic equation allowing a simple closed form solution. 
Finally, the obtained solution is written in a form which can be compared with experimental data collected by satellite hard X--ray spectrometer analyzing the sun, see e.g. \eqref{eq:continuity-solution}. In the last section we have shown that the changes suggested for the description of the motion of accelerated electrons are still in accordance with data recorded by the RHESSI spectrometer. More generally, it has been shown in \cite{emsbarbro01,torre, codi13, guo13} that the continuity equation can be used as a tool for the validation of theoretical models for the injection mechanisms and for the accelerated electrons energy loss using data provided by different imaging instruments. For this reason, having a more manageable equations is important, as it allows a more accurate analysis to choose among different acceleration models.

\section*{Acknowledgements}
The authors would like to thank M. Piana, G. Caviglia and G. Torre for useful discussions and for fruitful comments 
about the manuscript.\\

\end{document}